\documentclass[10pt,a4paper,reqno]{amsart}
\usepackage{amssymb,amsmath}
\usepackage{tikz}
\usepackage{longtable}
\usepackage{tabu}
\makeatletter 
\def\leqn{\tagsleft@true} 
\def\reqn{\tagsleft@false} 
\def\fleq{\@fleqntrue\let\mathindent\@mathmargin \@mathmargin=0pt} 
\def\cneq{\@fleqnfalse} 
\makeatother

\usepackage{amsthm}
\usepackage{latexsym}
\usepackage{mathrsfs}
\usepackage{tipa}
\usepackage{color}
\usepackage{xcolor}
\usepackage[numbers]{natbib}

\usepackage[linktocpage=true,
            pdfauthor={S. Aksteiner, T. Backdahl},
            pdftitle={A space-time calculus based on symmetric 2-spinors}]{hyperref}

\usepackage{orcidlink}

\hypersetup{
    colorlinks=true,
    linkcolor=blue,
    filecolor=blue,      
    urlcolor=blue,
}

\theoremstyle{plain}
\newtheorem{theorem}{Theorem}
\newtheorem{proposition}[theorem]{Proposition}
\newtheorem{lemma}[theorem]{Lemma}
\newtheorem{corollary}[theorem]{Corollary}

\newtheorem*{remark}{Remark}
\newtheorem{definition}[theorem]{Definition}

\def\sDiv{\mathscr{D}}
\def\sCurl{\mathscr{C}}
\def\sCurlDagger{\mathscr{C}^\dagger}
\def\sTwist{\mathscr{T}}

\newcommand{\SymDyadBasis}[4]{\mathcal{B}^{#1, #2}_{ #3, #4}}

\def\clap#1{\hbox to 0pt{\hss#1\hss}}

\DeclareMathOperator{\tho}{\text{\textthorn}}
\DeclareMathOperator{\edt}{\eth}

\def\SymSpace{\mathcal{S}}
\newcommand{\squareAB}{\square}

\newcounter{mnotecount}[section]

\author[S. Aksteiner]{Steffen Aksteiner \orcidlink{0000-0002-0009-4292}}
\email{steffen.aksteiner@aei.mpg.de}
\address{Albert Einstein Institute, Am M\"uhlenberg 1, D-14476 Potsdam, Germany }

\author[T. B\"{a}ckdahl]{Thomas B\"{a}ckdahl \orcidlink{0000-0003-3240-2445}}
\email{thomas.backdahl@chalmers.se}
\address{Mathematical Sciences, Chalmers University of Technology and University of Gothenburg, SE-412~96 Gothenburg, Sweden}

\newenvironment{mma}{%
\fleq%
\tabulinesep = 4pt%
\begin{equation}%
\begin{tabu} to 0.93\linewidth {@{\hspace{3ex}}r@{\hspace{2ex}}X@{}}%
}{%
\end{tabu}%
\end{equation}%
\cneq%
}%

\newenvironment{mma*}{%
\fleq%
\tabulinesep = 4pt%
\begin{equation*}%
\begin{tabu} to 0.93\linewidth {@{\hspace{3ex}}r@{\hspace{2ex}}X@{}}%
}{%
\end{tabu}%
\end{equation*}%
\cneq%
}%

\newcommand{\mmain}[1]{\mmaintext &  \texttt{\bfseries #1}}

\newcommand{\mmaout}[1]{\mmaouttext & ${#1}$ }
\newcommand{\mmanoout}[1]{\phantom{\mmaouttext} & ${#1}$ }

\newcommand{\mmaouttext}{\mmainouttext{Out\,=}}
\newcommand{\mmaintext}{\mmainouttext{\hspace{2ex}In\,:=}}
\newcommand{\mmainouttext}[1]{\textrm{\normalfont{\small\textcolor{darkgray}{#1}}}}

\newcommand{\ruledelayed}{\tikz[baseline=-\the\dimexpr\fontdimen22\textfont2\relax]{\filldraw (0.04,0.05) circle (0.6pt);	\filldraw (0.04,-0.05) circle (0.6pt); \draw[line width=0.8pt] (0.07,0) -- (0.25,0); \draw[line width=0.8pt] (0.16,0.08) -- (0.25,0) -- (0.16,-0.08);}\;}

\newcommand{\funcref}[1]{\hyperref[#1]{{\ttfamily\nameref*{#1}}}}

\newsavebox{\coloredbgbox}

\begin{document}

\title{A space-time calculus based on symmetric 2-spinors}
\allowdisplaybreaks[2]

\begin{abstract}
In this paper we present a space-time calculus for symmetric spinors, including a product with a number of index contractions followed by symmetrization. 
As all operations stay within the class of symmetric spinors, no involved index manipulations are needed. In fact spinor indices are not needed in the formalism. It is also general because any covariant tensor expression in a 4-dimensional Lorentzian spacetime can be translated to this formalism.
The computer algebra implementation \emph{SymSpin} as part of \emph{xAct} for \emph{Mathematica} is also presented.
\end{abstract}

\maketitle

\section{Introduction}
When working with tensorial expressions, one usually encounters difficulties handling index manipulations due to complicated symmetries. 
Techniques including group theoretical calculations and Young tableaux have been introduced to try to tackle these problems. However, their complexity grows quickly with the size of the problem.
The purpose of this paper is to present a formalism based on 2-spinors that aims to simplify the situation by utilizing the symmetry properties of irreducible spinors.

Let $(\mathcal{M}, g_{ab})$ be a 4-dimensional manifold with metric $g_{ab}$ of Lorentzian signature and admitting a spin structure with spin metric $\epsilon_{AB}$. It is well known that any tensor field on $\mathcal{M}$ can be expressed in terms of 2-spinors,  which in turn can be decomposed into irreducible symmetric spinors \cite[Prop 3.3.54]{PR1}. For instance a valence $(3,0)$ spinor can be decomposed as
\begin{align}
T_{ABC}={}&T_{(ABC)}
 + \tfrac{1}{3} T^{D}{}_{D(B}\epsilon_{C)A}
 -  \tfrac{1}{3} \epsilon_{A(B}T^{D}{}_{C)D}
 - \tfrac{1}{2} T_{A}{}^{D}{}_{D} \epsilon_{BC}.
\end{align}
Therefore, it is sufficient to work with with symmetric spinors. To fully establish this perspective, a symmetric product for symmetric spinors with a number of contractions is needed. It is the intention of this work to introduce the corresponding algebra and to derive its basic properties. 
In particular, with these operations we stay within the algebra of symmetric spinors. This offers great simplifications, and speeds up the calculations.
Furtheremore, no relevant information is left in the indices, and we therefore get an index-free compact formalism. 

We have previously described the decomposition of the covariant derivative \cite{AndBaeBlu14a}, leading to four fundamental spinor operators, which can be viewed as a special case. Also, the symmetric product is a generalization of some special operators, like the $\mathcal{K}^i$ operators defined in \cite[Definition II.4]{2016arXiv160106084A}. Therefore, all properties of such operators can easily be derived from the corresponding properties of the symmetric product described in this paper.

The formalism has many potential applications, see \cite{CompComplex},\cite{AdjointOp}. As a simple example, consider a condition of the form
\begin{align}
0={}&K_{AB}{}^{FH} L_{F}{}^{C} \varphi_{HC}
 + M_{(A}{}^{C}\varphi_{B)C},
\end{align}
for symmetric spinors $K, L, M, \varphi$. For arbitrary $\varphi$ a systematic computation, using the techniques of this paper, shows that the conditions on $K,L, M$ are of the form
\begin{align}
K^{G}{}_{(ABC}L_{|G|F)}={}&0,&
M_{AB}={}&\tfrac{1}{2} K^{CF}{}_{AB} L_{CF},
\end{align}
see Section~\ref{sec:example} for details. The same techniques have been used in \cite{JacBac2022} to derive conditions on the spacetime for the existence of second order symmetry operators for the massive Dirac equation.

The formalism is implemented in the \emph{SymSpin} \cite{SymSpinWeb} package for \emph{xAct} \cite{xActWeb} for \emph{Mathematica}.

In Section~\ref{sec:SymSpinAlg} we introduce the symmetric product and state basic properties in Theorem~\ref{thm:SymProdProperties}. The expansion of a product into symmetric products is discussed in Lemma~\ref{lem:IrrDecProduct}. The irreducible parts of the Levi-Civita connection, its commutators, curvature and Leibniz rules are discussed in Section~\ref{sec:Derivatives}. A concise form the the dyad components of such symmetric spinors is given in Section~\ref{sec:GHP}. The computer algebra implementation is discussed in Section~\ref{sec:SymSpinPackage} and Section~\ref{sec:Conclusions} contains some conclusions.

\section{Symmetric spinor algebra} \label{sec:SymSpinAlg} 
Let $\SymSpace_{k,l}$ be the space of symmetric valence $(k,l)$ spinors. In abstract index notation, elements are of the form $\phi_{A_1 \dots A_k A'_1 \dots A'_l} \in \SymSpace_{k,l}$. Sometimes it is convenient to suppress the valence and/or indices and we write e.g. $\phi \in \mathcal{S}$ or $\phi \in \mathcal{S}_{k,l}$.

\subsection{Symmetric product}
Given two symmetric spinors, we introduce a product which involves a given number of contractions and symmetrization afterwards.

\begin{definition} \label{def:SymProd}
Let $k,l,n,m,i,j$ be integers with $i \leq min(k,n)$ and $j \leq min(l,m)$. The symmetric product is a bilinear form 
\begin{align}
\overset{i,j}{\odot}: \SymSpace_{k,l} \times \SymSpace_{n,m} \to{}& \SymSpace_{k+n-2i,l+m-2j}.
\end{align}
For $\phi \in \mathcal{S}_{k,l}, \psi \in \mathcal{S}_{n,m}$, it is given by
\begin{align}
\label{eq:SymMultDef}
(\phi\overset{i,j}{\odot}\psi)_{A_1 \dots A_{k+n-2i}}^{A'_1 \dots A'_{l+m-2j}}={}& 
\phi_{(A_1 \dots A_{k-i-1}}^{(A'_1 \dots A'_{l-j-1}| B_1 \dots B_i B'_1 \dots B'_j|} 
\psi^{A'_{l-j} \dots A'_{l+m-2j})}_{A_{k-i} \dots A_{k+n-2i})B_1 \dots B_i B'_1 \dots B'_j}
\end{align}
\end{definition} 

For many commutator relations we will need the following coefficients.
\begin{definition}
Define the associativity coefficients
\begin{align} 
F_{i,r,k}^{t,m,M}={}&\sum_{p=0}^M\sum_{q=0}^{M-p}
\frac{(-1)^{t-p+q} \binom{k -  m}{p} \binom{m}{M -  p -  q} \binom{k -  m -  p}{q} \binom{r - m}{t - p} \binom{i -  t}{M -  p -  q} \binom{t - p}{q}}{\binom{i + k -  M -  p + 1}{M -  p} \binom{M -  p}{q} \binom{k - 2 m + r}{t}}.
\end{align}
\end{definition}
Observe that the limits can be restricted to $\max(0, m -  r + t) \leq p \leq \min(k -  m, M, t)$ and $\max(0, M - m -  p, M - i -  p + t) \leq q \leq \min(k -  m -  p, M -  p, t - p)$ because the terms are zero outside this range.

For multiple products we will use the convention $\omega\overset{m,n}{\odot}\varphi\overset{t,u}{\odot}\phi=\omega\overset{m,n}{\odot}(\varphi\overset{t,u}{\odot}\phi)$.

\begin{theorem} \label{thm:SymProdProperties}
Let $\phi \in \SymSpace_{i,j}, \omega \in \SymSpace_{r,s}, \varphi \in \mathcal{S}_{k,l}$. The symmetric product $\odot$ of Definition \ref{def:SymProd} has the following properties:
\begin{subequations} 
\begin{enumerate}
 \item \label{part:SymProdProp1} It is graded anti-commutative:
 \begin{align}
  \phi\overset{m,n}{\odot}\omega ={}& (-1)^{m+n}\omega\overset{m,n}{\odot}\phi \label{eq:Commutativity}
 \end{align}
 \item \label{part:SymProdProp2} It is non-associative:
\begin{align} \label{eq:NonAssociative}
(\omega\overset{m,n}{\odot}\varphi)\overset{t,u}{\odot}\phi
={}&\negmedspace\negmedspace\sum_{M=0}^{\min(i,k)}\sum_{N=0}^{\min(j,l)}\negmedspace (-1)^{t+u+M+N} F_{i,r,k}^{t,m,M}F_{j, s, l}^{u, n, N}
\omega\overset{t+m-M,u+n-N}{\odot}\varphi\overset{M,N}{\odot}\phi.
\end{align}
 \item \label{part:SymProdProp3} It is Hermitian:
 \begin{align}
  \overline{\phi\overset{m,n}{\odot}\omega} = \overline\phi\overset{n,m}{\odot}\overline\omega
 \end{align}
\end{enumerate}
\end{subequations} 
\end{theorem}
Combining the first two points, we get the following useful relation.
\begin{corollary} 
\begin{align}
\phi\overset{t,u}{\odot}\omega\overset{m,n}{\odot}\varphi
={}&\negmedspace\negmedspace\sum_{M=0}^{\min(i,k)}\sum_{N=0}^{\min(j,l)}\negmedspace F_{i,r,k}^{t,m,M}F_{j, s, l}^{u, n, N}
\omega\overset{t+m-M,u+n-N}{\odot}\phi\overset{M,N}{\odot}\varphi
\label{eq:SymMultCommutator1}
\end{align}
\end{corollary} 

\subsection{Irreducible decomposition}
A key property of the symmetric product is that the product of two symmetric spinors can always be decomposed in terms of symmetric products and spin metrics.
\begin{definition}
We will use the following notation for for products of spin metrics.
\begin{subequations}
\begin{align}
\epsilon_{A_1\dots A_p}^{B_1\dots B_p}&=\epsilon_{A_1}{}^{B_1}\dots \epsilon_{A_p}{}^{B_p},\\
\bar\epsilon_{A_1'\dots A_q'}^{B_1'\dots B_q'}&=\bar\epsilon_{A_1'}{}^{B_1'}\dots \bar\epsilon_{A_q'}{}^{B_q'}.
\end{align}
\end{subequations}
\end{definition}
\begin{lemma}\label{lem:IrrDecProduct}
For $\phi \in \SymSpace_{i,j}, \varphi \in \mathcal{S}_{k,l}$ with $p$ unprimed and $q$ primed contractions, we have the irreducible decomposition
\begin{align}
&\hspace{-3ex}\phi^{C_{1}\dots C_pC_1'\dots C_q'}_{A_1\dots A_{i-p}A_1'\dots A_{j-q}'}\varphi_{C_1\dots C_pC_1'\dots C_q'}^{B_1\dots B_{k-p} B_1'\dots B_{l-q}'}\nonumber\\
={}&(-1)^{p+q}\negmedspace\sum_{m=p}^{\min(i,k)}\sum_{n=q}^{\min(j,l)}\Bigl(
\frac{(-1)^{m+n}\binom{i -  p}{m -  p} \binom{k -  p}{m -  p}\binom{j-q}{n-q} \binom{l-q}{n-q}}{\binom{i + k -  m -  p + 1}{m -  p}\binom{j + l - n - q + 1}{n-q}} \nonumber\\
&\times \epsilon_{(A_1\dots A_{m-p}}^{(B_1\dots B_{m-p}}(\phi\overset{m,n}{\odot}\varphi){}_{A_{m-p+1}\dots A_{i-p})(A_{n-p+1}'\dots A_{j-q}'}^{B_{m-p+1}\dots B_{k-p})(B_{n-q+1}'\dots B_{l-q}'}\bar\epsilon_{A_1'\dots A_{n-q}')}^{B_1'\dots B_{n-q}')}
\Bigr).
\end{align}
\end{lemma}
\begin{proof}
Let $\phi$ and $\varphi$ be symmetric of valence $(i,0)$ and $(k,0)$ respectively. By \cite[Prop 3.3.54]{PR1} the irreducible decomposition of the product must have the following form
\begin{align}
\label{eq:IrrDecAnsatz}
\phi_{A_1\dots A_i}\varphi^{B_1\dots B_k}=\sum_{m=0}^{\min(i,k)} c_m \epsilon_{(A_1\dots A_m}^{(B_1\dots B_m}(\phi\overset{m,0}{\odot}\varphi){}_{A_{m+1}\dots A_i)}^{B_{m+1}\dots B_k)}
\end{align}
Taking a trace of the summand, we find by partial expansions of the symmetrizations that
\begin{align}
&\hspace{-3ex}\epsilon_{(A_1\dots A_m}^{(B_1\dots B_m}(\phi\overset{m,0}{\odot}\varphi){}_{A_{m+1}\dots A_i)}^{B_{m+1}\dots B_{k-1}A_i)}\nonumber\\
={}&\tfrac{m}{i}\epsilon_{A_i(A_1\dots A_{m-1}}^{(B_1\dots B_m}(\phi\overset{m,0}{\odot}\varphi){}_{A_m\dots A_{i-1})}^{B_{m+1}\dots B_{k-1}A_i)}
+\tfrac{i-m}{i}\epsilon_{(A_1\dots A_{m}}^{(B_1\dots B_m}(\phi\overset{m,0}{\odot}\varphi){}_{A_{m+1}\dots A_{i-1})A_i}^{B_{m+1}\dots B_{k-1}A_i)}\nonumber\\
={}&\tfrac{m}{ik}\epsilon_{A_i(A_1\dots A_{m-1}}^{A_i(B_1\dots B_{m-1}}(\phi\overset{m,0}{\odot}\varphi){}_{A_{m}\dots A_{i-1})}^{B_{m}\dots B_{k-1})}
+\tfrac{m(m-1)}{ik}\epsilon_{A_i(A_1\dots A_{m-1}}^{(B_1|A_i|\dots B_{m-1}}(\phi\overset{m,0}{\odot}\varphi){}_{A_{m}\dots A_{i+1})}^{B_{m}\dots B_{k-1})}\nonumber\\
&+\tfrac{m(k-m)}{ik}\epsilon_{A_i(A_1\dots A_{m-1}}^{(B_1\dots B_m}(\phi\overset{m,0}{\odot}\varphi){}_{A_{m}\dots A_{i-1})}^{B_{m+1}\dots B_{k-1})A_i}
+\tfrac{(i-m)m}{ik}\epsilon_{(A_1\dots A_{m}}^{A_i(B_1\dots B_{m-1}}(\phi\overset{m,0}{\odot}\varphi){}_{A_{m+1}\dots A_{i-1})A_i}^{B_{m}\dots B_{k-1})}\nonumber\\
={}&\tfrac{m(i+k-m+1)}{ik}
\epsilon_{(A_1\dots A_{m-1}}^{(B_1\dots B_{m-1}}(\phi\overset{m,0}{\odot}\varphi){}_{A_{m}\dots A_{i-1})}^{B_{m}\dots B_{k-1})}.
\end{align}
Recursively for $p\leq \min(i,k)$ traces we get
\begin{align}
&\hspace{-3ex}\epsilon_{(A_1\dots A_m}^{(B_1\dots B_m}(\phi\overset{m,0}{\odot}\varphi){}_{A_{m+1}\dots A_i)}^{B_{m+1}\dots B_{k-p}A_{i-p+1}\dots A_i)}\nonumber\\
={}&\tfrac{m(i+k-m+1)}{ik}
\epsilon_{(A_1\dots A_{m-1}}^{(B_1\dots B_{m-1}}(\phi\overset{m,0}{\odot}\varphi){}_{A_{m}\dots A_{i-1})}^{B_{m}\dots  B_{k-p}A_{i-p+1}\dots A_{i-1})}\nonumber\\
={}&\tfrac{m(i+k-m+1)}{ik}\tfrac{(m-1)(i+k-m)}{(i-1)(k-1)}
\epsilon_{(A_1\dots A_{m-2}}^{(B_1 \dots B_{m-2}}(\phi\overset{m,0}{\odot}\varphi){}_{A_{m-1}\dots A_{i-2})}^{B_{m-1}\dots  B_{k-p}A_{i-p+1}\dots A_{i-2})}\nonumber\\
={}&
\epsilon_{(A_1\dots A_{m-p}}^{(B_1\dots B_{m-p}}(\phi\overset{m,0}{\odot}\varphi){}_{A_{m-p+1}\dots A_{i-p})}^{B_{m-p+1}\dots  B_{k-p})}\prod_{q=0}^{p-1}\tfrac{(m-q)(i+k-m+1-q)}{(i-q)(k-q)}.
\nonumber\\
={}&
\epsilon_{(A_1\dots A_{m-p}}^{(B_1\dots B_{m-p}}(\phi\overset{m,0}{\odot}\varphi){}_{A_{m-p+1}\dots A_{i-p})}^{B_{m-p+1}\dots  B_{k-p})}
\frac{\binom{1 + i + k -  m}{p} \binom{m}{p}}{\binom{i}{p} \binom{k}{p}}
\end{align}
Taking $p\leq \min(i,k)$ traces in \eqref{eq:IrrDecAnsatz} gives
\begin{align}
&\hspace{-3ex}\phi_{A_1\dots A_{i-p}}^{C_1\dots C_p}\varphi^{B_1\dots B_{k-p}}_{C_1\dots C_p}\nonumber\\
={}&(-1)^p\phi_{A_1\dots A_i}\varphi^{B_1\dots B_{k-p}A_{i-p+1}\dots A_i}\nonumber\\
={}&(-1)^p\sum_{m=0}^{\min(i,k)} c_m \epsilon_{(A_1\dots A_m}^{(B_1\dots B_m}(\phi\overset{m,0}{\odot}\varphi){}_{A_{m+1}\dots A_i)}^{B_{m+1}B_{k-p}A_{i-p+1}\dots A_i)}\nonumber\\
={}&(-1)^p\sum_{m=p}^{\min(i,k)}  c_m\frac{\binom{i + k -  m +1}{p} \binom{m}{p}}{\binom{i}{p} \binom{k}{p}}\epsilon_{(A_1\dots A_{m-p}}^{(B_1\dots B_{m-p}}(\phi\overset{m,0}{\odot}\varphi){}_{A_{m-p+1}\dots A_{i-p})}^{B_{m-p+1}\dots  B_{k-p})}.
\end{align}
With $m<p$, we get at least one contraction of the symmetric spinor $(\phi\overset{m,0}{\odot}\varphi)$ and the term drops out.
If we symmetrize over all free indices, only the $m=p$ term survives, and we get
\begin{align}
c_m =\frac{(-1)^m\binom{i}{m} \binom{k}{m}}{\binom{i + k -  m +1}{m}}.
\end{align}
Hence
\begin{align}
&\hspace{-3ex}\phi_{A_1\dots A_{i-p}}^{C_1\dots C_p}\varphi^{B_1\dots B_{k-p}}_{C_1\dots C_p}\nonumber\\
={}&(-1)^p\sum_{m=p}^i  \frac{(-1)^m\binom{i}{m} \binom{k}{m}}{\binom{i + k -  m + 1}{m}}\frac{\binom{i + k -  m +1}{p} \binom{m}{p}}{\binom{i}{p} \binom{k}{p}}\epsilon_{(A_1\dots A_{m-p}}^{(B_1\dots B_{m-p}}(\phi\overset{m,0}{\odot}\varphi){}_{A_{m-p+1}\dots A_{i-p})}^{B_{m-p+1}\dots  B_{k-p})}\nonumber\\
={}&(-1)^p\sum_{m=p}^{\min(i,k)} \frac{(-1)^m\binom{i -  p}{m -  p} \binom{k -  p}{m -  p}}{\binom{i + k -  m -  p +1}{m -  p}}\epsilon_{(A_1\dots A_{m-p}}^{(B_1\dots B_{m-p}}(\phi\overset{m,0}{\odot}\varphi){}_{A_{m-p+1}\dots A_{i-p})}^{B_{m-p+1}\dots  B_{k-p})}.
\end{align}
By complex conjugation we get the corresponding decomposition for the primed indices. 
\end{proof}

\subsection{Proof of Theorem~\ref{thm:SymProdProperties}}
To proof the main theorem and in particular \eqref{eq:NonAssociative}, we need the following intermediate identities. We restrict to unprimed indices, as the effect of primed indices can be superimposed.
We begin with a partial expansion of symmetrization of $B$ indices.
\begin{proposition}\label{prop:PartialSymMultExpansion}
Let $\omega \in \SymSpace_{r,0}, \varphi \in \mathcal{S}_{k,0}$. We have the partial expansion
\begin{align}
\label{eq:PartialSymMultExpansion}
&\hspace{-3ex}(\omega\overset{m,0}{\odot}\varphi){}_{A_1\dots A_{k+r-2m-t}B_1\dots B_t}\nonumber\\
={}&\sum_{p=0}^{t}\frac{\binom{k -  m}{p} \binom{r - m}{t-p}}{\binom{k - 2 m + r}{t}}
 \omega_{B_{p+1}\cdots B_{t}(A_1\dots A_{r-m-t+p}}^{C_1\dots C_m}\varphi_{A_{r-m-t+p+1}\dots A_{k+r-2m-t})B_1\dots B_pC_1\dots C_m}.
\end{align}
The sum can be limited to the range $\max(0,t+m-r)\leq p\leq \min(t,k-m)$.
\end{proposition}
\begin{proof}
Partial expansion of the symmetry for the indices $B_t, B_{t-1}, \dots B_1$ gives
\begin{align}
&\hspace{-3ex}(\omega\overset{m,0}{\odot}\varphi){}_{A_1\dots A_{k+r-2m-t}B_1\dots B_t}\nonumber\\
={}&\tfrac{r-m}{k+r-2m}\omega_{B_t(A_1\dots A_{r-m-1}}^{C_1\dots C_m}\varphi_{A_{r-m}\dots A_{k+r-2m-t}B_1\dots B_{t-1})C_1\dots C_m}\nonumber\\
&+\tfrac{k-m}{k+r-2m}\omega_{(A_1\dots A_{r-m}}^{C_1\dots C_m}\varphi_{A_{r-m+1}\dots A_{k+r-2m-t}B_1\dots B_{t-1})B_tC_1\dots C_m}\nonumber\\
={}&\sum_{p=0}^{t}\Bigl(\binom{t}{p}\tfrac{(r-m)!(k-m)!(k+r-2m-t)!}{(r-m-t+p)!(k-m-p)!(k+r-2m)!}\nonumber\\
&\times \omega_{B_{p+1}\cdots B_{t}(A_1\dots A_{r-m-t+p}}^{C_1\dots C_m}\varphi_{A_{r-m-t+p+1}\dots A_{k+r-2m-t})B_1\dots B_pC_1\dots C_m}\Bigr),
\end{align}
which can be simplified to \eqref{eq:PartialSymMultExpansion}.
\end{proof}

We aso need to make an irreducible decomposition of a product of two spinors with some contractions and symmetrizations.
\begin{proposition}\label{prop:PropEight}
Let $\phi \in \SymSpace_{i,0}, \varphi \in \mathcal{S}_{k,0}$.
\begin{align}
&\hspace{-3ex}\phi^{C_{1}\dots C_p(B_1\dots B_{t-p}}_{(A_1\dots A_{i-t}}\varphi_{A_{i-t+1}\dots A_{i+k-t-m-p})C_1\dots C_p}^{B_{t-p+1}\dots B_{t-p+m})}\nonumber\\
={}&\negmedspace\sum_{M=0}^{\min(i,k)-p} \sum_{q=0}^M\frac{(-1)^{q+M}\binom{m}{M -  q} \binom{k -  m -  p}{q} \binom{i -  t}{M -  q} \binom{t - p}{q}}{\binom{i + k -  M - 2p + 1}{M}\binom{M}{q} }
\epsilon_{(A_1\dots A_{M}}^{(B_1\dots B_M}(\phi\overset{M+p,0}{\odot}\varphi){}_{A_{M+1}\dots A_{i+k-t-m-p})}^{B_{M+1}\dots  B_{t-p+m})}
\label{eq:PropEight}
\end{align}
\end{proposition}

\begin{proof}
Let $\bumpeq$ mean equal after lowering the $A$ indices, raising the $B$ indices and symmetrizing over the $A$ and $B$ index sets separately.
Using Lemma~\ref{lem:IrrDecProduct}, performing a partial expansion of the symmetries and noticing that $\epsilon_{A_i}^{A_j}\bumpeq 0$ and $\epsilon_{B_i}^{B_j}\bumpeq 0$ if $i\neq j$, we get
\begin{align}
&\hspace{-3ex}\phi^{C_{1}\dots C_p}_{A_1\dots A_{i-t}B_1\dots B_{t-p}}\varphi_{C_1\dots C_p}^{A_{i-t+1}\dots A_{i+k-t-m-p}B_{t-p+1}\dots B_{t-p+m}}\nonumber\\
={}&(-1)^{p}\negmedspace\sum_{M=0}^{\min(i,k)-p}\Bigl(
\frac{(-1)^{M+p}\binom{i - p}{M} \binom{k -  p}{M}}{\binom{i + k -  M - 2p + 1}{M}} \nonumber\\
&\times \epsilon_{(A_1\dots A_{M}}^{(A_{i-t+1}\dots A_{i-t+M}}(\phi\overset{M+p,0}{\odot}\varphi){}_{A_{M+1}\dots A_{i-t}B_1\dots B_{t-p})}^{A_{i-t+M+1}\dots  A_{i+k-t-m-p}B_{t-p+1}\dots B_{t-p+m})}
\Bigr)\nonumber\\
\bumpeq{}&\negmedspace\sum_{M=0}^{\min(i,k)-p}\frac{(-1)^{M}\binom{i - p}{M} \binom{k -  p}{M}}{\binom{i + k -  M - 2p + 1}{M}}\Bigl( \nonumber\\
&\tfrac{(t-p)(k-m-p)}{(i-p)(k-p)} \epsilon_{B_1(A_1\dots A_{M-1}}^{A_{i-t+1}(A_{i-t+2}\dots A_{i-t+M}}(\phi\overset{M+p,0}{\odot}\varphi){}_{A_{M}\dots A_{i-t}B_2\dots B_{t-p})}^{A_{i-t+M+1}\dots  A_{i+k-t-m-p}B_{t-p+1}\dots B_{t-p+m})}\nonumber\\
&+\tfrac{(i-t)m}{(i-p)(k-p)} \epsilon_{A_1(A_2\dots A_{M}}^{B_{t-p+1}(A_{i-t+1}\dots A_{i-t+M-1}}(\phi\overset{M+p,0}{\odot}\varphi){}_{A_{M+1}\dots A_{i-t}B_1\dots B_{t-p})}^{A_{i-t+M}\dots  A_{i+k-t-m-p}B_{t-p+2}\dots B_{t-p+m})}
\Bigr).
\end{align}
Repeatedly expanding, we find 
\begin{align}
&\hspace{-3ex}\phi^{C_{1}\dots C_p}_{A_1\dots A_{i-t}B_1\dots B_{t-p}}\varphi_{C_1\dots C_p}^{A_{i-t+1}\dots A_{i+k-t-m-p}B_{t-p+1}\dots B_{t-p+m}}\nonumber\\
\bumpeq{}&\negmedspace\sum_{M=0}^{\min(i,k)-p}\frac{(-1)^{M}\binom{i - p}{M} \binom{k -  p}{M}}{\binom{i + k -  M - 2p + 1}{M}}\Bigl(\sum_{q=0}^M\binom{M}{q}\tfrac{(t-p)!(k-m-p)!(i-t)!m!(i-p-M)!(k-p-M)!}{(t-p-q)!(k-m-p-q)!(i-t-M+q)!(m-M+q)!(i-p)!(k-p)!}
\nonumber\\
&\epsilon_{B_1\dots B_qA_1\dots A_{M-q}}^{A_{i-t+1}\dots A_{i-t+q}B_{t-p+1}\dots B_{t-p+M-q}}(\phi\overset{M+p,0}{\odot}\varphi){}_{A_{M-q+1}\dots A_{i-t}B_{q+1}\dots B_{t-p}}^{A_{i-t+q+1}\dots  A_{i+k-t-m-p}B_{t-p+M-q+1}\dots B_{t-p+m}}
\Bigr).
\end{align}
Moving the $A$ indices down and the $B$ indices up and writing out the symmetrizations, we get 
\begin{align}
&\hspace{-3ex}\phi^{C_{1}\dots C_p(B_1\dots B_{t-p}}_{(A_1\dots A_{i-t}}\varphi_{A_{i-t+1}\dots A_{i+k-t-m-p})C_1\dots C_p}^{B_{t-p+1}\dots B_{t-p+m})}\nonumber\\
={}&\negmedspace\sum_{M=0}^{\min(i,k)-p}\frac{(-1)^{M}\binom{i - p}{M} \binom{k -  p}{M}}{\binom{i + k -  M - 2p + 1}{M}}\Bigl( \sum_{q=0}^M\frac{(-1)^q\binom{m}{M -  q} \binom{k -  m -  p}{q} \binom{i -  t}{M -  q} \binom{t - p}{q}}{\binom{M}{q} \binom{i -  p}{M} \binom{k -  p}{M}}
\nonumber\\
&\epsilon_{(A_{i-t+1}\dots A_{i-t+q}A_1\dots A_{M-q}}^{(B_1\dots B_q B_{t-p+1}\dots B_{t-p+M-q}}(\phi\overset{M+p,0}{\odot}\varphi){}_{A_{M-q+1}\dots A_{i-t}A_{i-t+q+1}\dots  A_{i+k-t-m-p})}^{B_{q+1}\dots B_{t-p}B_{t-p+M-q+1}\dots B_{t-p+m})}
\Bigr).
\end{align}
After rearranging the indices, and simplifying, we get \eqref{eq:PropEight}.
\end{proof}

\begin{proof}[Proof of Theorem \ref{thm:SymProdProperties}]
Part~\ref{part:SymProdProp1} follows from the zee-zaw rule on the $m+n$ contracted indices and part~\ref{part:SymProdProp3} follows from complex conjugation of \eqref{eq:SymMultDef}.
Part~\ref{part:SymProdProp2} follows from the following argument.
Proposition~\ref{prop:PartialSymMultExpansion}, a renaming of the contracted indices and using the zee-zaw rule gives
\begin{align}
&\hspace{-3ex}(\phi\overset{t,0}{\odot}\omega\overset{m,0}{\odot}\varphi){}_{A_1\dots A_{i+k+r-2m-2t}}\nonumber\\
={}&\sum_{p=0}^{t}\frac{\binom{k -  m}{p} \binom{r - m}{t-p}}{\binom{k - 2 m + r}{t}}\nonumber\\
&\times  \omega_{B_{p+1}\cdots B_{t}(A_1\dots A_{r-m-t+p}}^{C_1\dots C_m}\phi^{B_1\dots B_t}_{A_{k+r-2m-t+1}\dots A_{i+k+r-2m-2t}}\varphi_{A_{r-m-t+p+1}\dots A_{k+r-2m-t})B_1\dots B_pC_1\dots C_m} \nonumber\\
={}&\sum_{p=0}^{t}(-1)^{m}\frac{\binom{k -  m}{p} \binom{r - m}{t-p}}{\binom{k - 2 m + r}{t}}\nonumber\\
&\times  \omega_{B_1\dots B_{m+t-p}(A_{i-t+k-m-p+1}\dots A_{i+k+r-2m-2t}}\phi^{C_1\dots C_p B_{1} B_{t-p}}_{A_{1}\dots A_{i-t}}\varphi_{A_{i-t+1}\dots A_{i-t+k-m-p}) C_1\dots C_p}^{B_{t-p+1}\dots B_{t-p+m}}.
\end{align} 
Using Proposition~\ref{prop:PropEight}, contracting the spin metrics, and using the zee-zaw rule, we get
\begin{align}
&\hspace{-3ex}(\phi\overset{t,0}{\odot}\omega\overset{m,0}{\odot}\varphi){}_{A_1\dots A_{i+k+r-2m-2t}}\nonumber\\
={}&\sum_{p=0}^{t}\sum_{M=0}^{\min(i,k)-p} \sum_{q=0}^M(-1)^{m}\frac{\binom{k -  m}{p} \binom{r - m}{t-p}}{\binom{k - 2 m + r}{t}}\frac{(-1)^{q+M}\binom{m}{M -  q} \binom{k -  m -  p}{q} \binom{i -  t}{M -  q} \binom{t - p}{q}}{\binom{i + k -  M - 2p + 1}{M}\binom{M}{q} }\nonumber\\
&\times  \omega_{B_1\dots B_{m+t-p}(A_{i-t+k-m-p+1}\dots A_{i+k+r-2m-2t}}
\epsilon_{A_1\dots A_{M}}^{B_1\dots B_M}(\phi\overset{M+p,0}{\odot}\varphi){}_{A_{M+1}\dots A_{i+k-t-m-p})}^{B_{M+1}\dots  B_{t-p+m}}\nonumber\\
={}&\sum_{p=0}^{t}\sum_{M=0}^{\min(i,k)-p} \sum_{q=0}^M(-1)^{m+q+M}\frac{\binom{k -  m}{p} \binom{r - m}{t-p}\binom{m}{M -  q} \binom{k -  m -  p}{q} \binom{i -  t}{M -  q} \binom{t - p}{q}}{\binom{k - 2 m + r}{t}\binom{i + k -  M - 2p + 1}{M}\binom{M}{q} }\nonumber\\
&\times  \omega_{B_{M+1}\dots B_{m+t-p}(A_1\dots A_{r+M+p-m-t}}
(\phi\overset{M+p,0}{\odot}\varphi){}_{A_{r+M+p-m-t+1}\dots A_{i+k+r-2m-2t})}^{B_{M+1}\dots  B_{t-p+m}}\nonumber\\
={}&\sum_{p=0}^{t}\sum_{M=0}^{\min(i,k)-p} \sum_{q=0}^M(-1)^{q+t-p}\frac{\binom{k -  m}{p} \binom{r - m}{t-p}\binom{m}{M -  q} \binom{k -  m -  p}{q} \binom{i -  t}{M -  q} \binom{t - p}{q}}{\binom{k - 2 m + r}{t}\binom{i + k -  M - 2p + 1}{M}\binom{M}{q} }\nonumber\\
&\times  \omega^{B_{M+1}\dots B_{m+t-p}}_{(A_1\dots A_{r+M+p-m-t}}
(\phi\overset{M+p,0}{\odot}\varphi){}_{A_{r+M+p-m-t+1}\dots A_{i+k+r-2m-2t})B_{M+1}\dots  B_{t-p+m}}.
\end{align} 
Hence
\begin{align}
&\hspace{-3ex}(\phi\overset{t,0}{\odot}\omega\overset{m,0}{\odot}\varphi)\nonumber\\
={}&\negmedspace\sum_{p=0}^{t}\negmedspace\sum_{M=0}^{\min(i,k)-p} \sum_{q=0}^M
\frac{(-1)^{t-p+q}\binom{k -  m}{p} \binom{r - m}{t-p}\binom{m}{M -  q} \binom{k -  m -  p}{q} \binom{i -  t}{M -  q} \binom{t - p}{q}}{\binom{k - 2 m + r}{t}\binom{i + k -  M - 2p + 1}{M}\binom{M}{q}}
(\omega\overset{t+m-p-M,0}{\odot}\phi\overset{M+p,0}{\odot}\varphi)\nonumber\\
={}&\negmedspace\sum_{p=0}^{t}\negmedspace\sum_{M=p}^{\min(i,k)} \sum_{q=0}^{M-p}
\frac{(-1)^{t-p+q}\binom{k -  m}{p} \binom{r - m}{t-p}\binom{m}{M - p -  q} \binom{k -  m -  p}{q} \binom{i -  t}{M - p -  q} \binom{t - p}{q}}{\binom{k - 2 m + r}{t}\binom{i + k -  M - p + 1}{M-p}\binom{M-p}{q}}
(\omega\overset{t+m-M,0}{\odot}\phi\overset{M,0}{\odot}\varphi)\nonumber\\
={}&\negmedspace\sum_{M=0}^{\min(i,k)}\negmedspace F_{i,r,k}^{t,m,M}
(\omega\overset{t+m-M,0}{\odot}\phi\overset{M,0}{\odot}\varphi),
\end{align}
where we have made the change $M\rightarrow M-p$ and re-ordered the sums.
The limits can be restricted to $\max(0, m -  r + t) \leq p \leq \min(k -  m, M, t)$ and $\max(0, M - m -  p, M - i -  p + t) \leq q \leq \min(k -  m -  p, M -  p, t - p)$ because the terms are zero outside this range.
The treatment of the primed indices is completely analogous. 
\end{proof}

\subsection{Derivatives} \label{sec:Derivatives}
In \cite{AndBaeBlu14a}, the irreducible decomposition of the covariant derivative of a symmetric spinor was done in terms of fundamental spinor operators. 
By extending the symmetric product to the space of linear, symmetric differential operators of valence $(k,l)$, $\mathcal{O}_{k,l}$, we can express the fundamental spinor operators in a compact way.

\begin{remark}
For $\nabla \in \mathcal{O}_{1,1}$ we have the fundamental spinor operators \cite[Definition 13]{AndBaeBlu14a}
\begin{align}  \label{eq:FundSpinOps} 
\sDiv \varphi ={}& \nabla \overset{1,1}{\odot}\varphi, &&&
\sCurl \varphi ={}& \nabla \overset{0,1}{\odot}\varphi, &&&
\sCurlDagger \varphi ={}& \nabla \overset{1,0}{\odot}\varphi, &&&
\sTwist \varphi ={}& \nabla \overset{0,0}{\odot}\varphi.
\end{align}
\end{remark}

On $\varphi \in \mathcal{S}_{k,l}$ we have the irreducible decomposition of the covariant derivative into fundamental operators \cite[Lemma 15]{AndBaeBlu14a},
\begin{align}
\nabla_{A_1}{}^{A_1'}\varphi{}_{A_2\dots A_{k+1}}{}^{A_2'\dots A_{l+1}'}={}&
(\sTwist\varphi){}_{A_1\dots A_{k+1}}{}^{A_1'\dots A_{l+1}'}\nonumber\\
&-\tfrac{l}{l+1}\bar\epsilon^{A_1'(A_2'}(\sCurl\varphi){}_{A_1\dots A_{k+1}}{}^{A_3'\dots A_{l+1}')}\nonumber\\
&-\tfrac{k}{k+1}\epsilon_{A_1(A_2}(\sCurlDagger\varphi){}_{A_3\dots A_{k+1})}{}^{A_1'\dots A_{l+1}'}\nonumber\\
&+\tfrac{kl}{(k+1)(l+1)}\epsilon_{A_1(A_2}\bar\epsilon^{A_1'(A_2'}(\sDiv\varphi){}_{A_3\dots A_{k+1})}{}^{A_3'\dots A_{l+1}')}.\label{eq:IrrDecGeneralDer}
\end{align}

Next, we write the commutators in the new notation.
Define the operator 
\begin{align}
\squareAB &= -(\nabla \overset{0,1}{\odot} \nabla) \in \mathcal{O}_{2,0},
\end{align}
and its complex conjugate $\overline{\squareAB} \in \mathcal{O}_{0,2}$.

In index notation, it reads
$\squareAB_{AB} = \nabla_{(A|A'|}\nabla_{B)}{}^{A'}$. Acting on $\varphi \in \mathcal{S}_{k,l}$ it can be expressed in terms of curvature via
\begin{subequations} 
\begin{align}
 \squareAB \overset{0,0}{\odot} \varphi ={}& - k \Psi \overset{1,0}{\odot} \varphi - l  \Phi \overset{0,1}{\odot} \varphi,\\
 \squareAB \overset{1,0}{\odot} \varphi ={}&  -(k-1) \Psi \overset{2,0}{\odot} \varphi - l  \Phi \overset{1,1}{\odot} \varphi + (k+2) \Lambda \overset{0,0}{\odot} \varphi, \\
 \squareAB \overset{2,0}{\odot} \varphi ={}& - (k-2) \Psi \overset{3,0}{\odot} \varphi - l  \Phi \overset{2,1}{\odot} \varphi.
\end{align}
\end{subequations} 
\begin{lemma}{\cite[Lemma 18]{AndBaeBlu14a}}\label{lemma:commutators}
Let $\varphi \in \mathcal{S}_{k,l}$. The operators $\sDiv$, $\sCurl$, $\sCurlDagger$ and $\sTwist$ satisfy the commutator relations
\begin{subequations}
\begin{align}
\sDiv \sCurl \varphi
 ={}&\tfrac{k}{k+1} \sCurl \sDiv \varphi - \overline{\squareAB} \overset{0,2}{\odot} \varphi,
 & k\geq 0, l\geq 2,  \label{eq:DivCurl}\\
\sDiv \sCurlDagger \varphi
 ={}&\tfrac{l}{l+1} \sCurlDagger \sDiv \varphi
- \squareAB \overset{2,0}{\odot} \varphi,
& k\geq 2, l\geq 0,\label{eq:DivCurlDagger}\\
\sCurl \sTwist \varphi
={}&\tfrac{l}{l+1} \sTwist \sCurl \varphi
- \squareAB \overset{0,0}{\odot} \varphi,
& k\geq 0, l\geq 0,\label{eq:CurlTwist}\\
\sCurlDagger \sTwist \varphi
={}&
\tfrac{k}{k+1} \sTwist \sCurlDagger \varphi
- \overline\squareAB \overset{0,0}{\odot} \varphi,
& k\geq 0, l\geq 0,\label{eq:CurlDaggerTwist}\\
\sDiv \sTwist \varphi
={}&
-(\tfrac{1}{k+1}+\tfrac{1}{l+1}) \sCurl \sCurlDagger \varphi
+\tfrac{l(l+2)}{(l+1)^2} \sTwist \sDiv \varphi 
-\tfrac{l+2}{l+1}\squareAB  \overset{1,0}{\odot} \varphi
-\tfrac{l}{l+1}\overline\squareAB  \overset{0,1}{\odot} \varphi,
 & k\geq 1, l\geq 0,\label{eq:DivTwistCurlCurlDagger}\\
\sDiv \sTwist \varphi
={}&
-(\tfrac{1}{k+1}+\tfrac{1}{l+1})\sCurlDagger \sCurl \varphi
+\tfrac{k(k+2)}{(k+1)^2}\sTwist \sDiv \varphi
-\tfrac{k}{k+1}\squareAB  \overset{1,0}{\odot} \varphi
-\tfrac{k+2}{k+1}\overline\squareAB  \overset{0,1}{\odot} \varphi,
& k\geq 0, l\geq 1,\label{eq:DivTwistCurlDaggerCurl}\\
\sCurl \sCurlDagger \varphi
={}&
\sCurlDagger \sCurl \varphi
+(\tfrac{1}{k+1}-\tfrac{1}{l+1})\sTwist \sDiv \varphi
-\squareAB  \overset{1,0}{\odot} \varphi
+\overline\squareAB  \overset{0,1}{\odot} \varphi,
& k\geq 1, l\geq 1.\label{eq:CurlCurlDagger}
\end{align}
\end{subequations}
\end{lemma}

\begin{lemma}
For symmetric spinors $\phi \in \SymSpace_{i,j}, \varphi \in \mathcal{S}_{k,l}$ we have the following Leibniz rules.
\begin{subequations}
\label{eq:SymMultLeibniz}
\begin{align}
\sTwist (\phi {\overset{m,n}{\odot}}\varphi)={}&(-1)^{m + n}\varphi {\overset{m,n}{\odot}}\sTwist \phi
 + \tfrac{(-1)^{m + n} n}{j + 1}\varphi {\overset{m,n - 1}{\odot}}\sCurl \phi
 + \tfrac{(-1)^{m + n} m}{i + 1}\varphi {\overset{m - 1,n}{\odot}}\sCurlDagger \phi\nonumber\\
& + \tfrac{(-1)^{m + n} m n}{(i + 1) (j + 1)}\varphi {\overset{m - 1,n - 1}{\odot}}\sDiv \phi
 + \phi {\overset{m,n}{\odot}}\sTwist \varphi
 + \tfrac{n}{l + 1}\phi {\overset{m,n - 1}{\odot}}\sCurl \varphi\nonumber\\
& + \tfrac{m}{k + 1}\phi {\overset{m - 1,n}{\odot}}\sCurlDagger \varphi
 + \tfrac{m n}{k l + k + l + 1}\phi {\overset{m - 1,n - 1}{\odot}}\sDiv \varphi 
\label{TwistLeibniz},\\
\sCurl (\phi {\overset{m,n}{\odot}}\varphi)={}&\tfrac{(-1)^{m + n + 1} (l -  n)}{j + l - 2 n}\varphi {\overset{m,n + 1}{\odot}}\sTwist \phi
 + \tfrac{(-1)^{m + n} (j -  n) (j + l -  n + 1)}{(j + 1) (j + l - 2 n)}\varphi {\overset{m,n}{\odot}}\sCurl \phi\nonumber\\
& + \tfrac{(-1)^{m + n + 1} m (l -  n)}{(i + 1) (j + l - 2 n)}\varphi {\overset{m - 1,n + 1}{\odot}}\sCurlDagger \phi\nonumber\\
& + \tfrac{(-1)^{m + n} m (j -  n) (j + l -  n + 1)}{(i + 1) (j + 1) (j + l - 2 n)}\varphi {\overset{m - 1,n}{\odot}}\sDiv \phi
 -  \tfrac{j -  n}{j + l - 2 n}\phi {\overset{m,n + 1}{\odot}}\sTwist \varphi\nonumber\\
& + \tfrac{(l -  n) (j + l -  n + 1)}{(l + 1) (j + l - 2 n)}\phi {\overset{m,n}{\odot}}\sCurl \varphi
 + \tfrac{m (- j + n)}{(k + 1) (j + l - 2 n)}\phi {\overset{m - 1,n + 1}{\odot}}\sCurlDagger \varphi\nonumber\\
& + \tfrac{m (l -  n) (j + l -  n + 1)}{(k + 1) (l + 1) (j + l - 2 n)}\phi {\overset{m - 1,n}{\odot}}\sDiv \varphi 
\label{CurlLeibniz},\\
\sCurlDagger (\phi {\overset{m,n}{\odot}}\varphi)={}&\tfrac{(-1)^{m + n + 1} (k -  m)}{i + k - 2 m}\varphi {\overset{m + 1,n}{\odot}}\sTwist \phi
 + \tfrac{(-1)^{m + n + 1} n (k -  m)}{(j + 1) (i + k - 2 m)}\varphi {\overset{m + 1,n - 1}{\odot}}\sCurl \phi\nonumber\\
& + \tfrac{(-1)^{m + n} (i -  m) (i + k -  m + 1)}{(i + 1) (i + k - 2 m)}\varphi {\overset{m,n}{\odot}}\sCurlDagger \phi\nonumber\\
& + \tfrac{(-1)^{m + n} n (i -  m) (i + k -  m + 1)}{(i + 1) (j + 1) (i + k - 2 m)}\varphi {\overset{m,n - 1}{\odot}}\sDiv \phi
 -  \tfrac{i -  m}{i + k - 2 m}\phi {\overset{m + 1,n}{\odot}}\sTwist \varphi\nonumber\\
& + \tfrac{n (- i + m)}{(l + 1) (i + k - 2 m)}\phi {\overset{m + 1,n - 1}{\odot}}\sCurl \varphi
 + \tfrac{(k -  m) (i + k -  m + 1)}{(k + 1) (i + k - 2 m)}\phi {\overset{m,n}{\odot}}\sCurlDagger \varphi\nonumber\\
& + \tfrac{n (k -  m) (i + k -  m + 1)}{(k + 1) (l + 1) (i + k - 2 m)}\phi {\overset{m,n - 1}{\odot}}\sDiv \varphi 
\label{CurlDgLeibniz},\\
\sDiv (\phi {\overset{m,n}{\odot}}\varphi)={}&\tfrac{(-1)^{m + n} (k -  m) (l -  n)}{(i + k - 2 m) (j + l - 2 n)}\varphi {\overset{m + 1,n + 1}{\odot}}\sTwist \phi\nonumber\\
& + \tfrac{(-1)^{m + n + 1} (j -  n) (k -  m) (j + l -  n + 1)}{(j + 1) (i + k - 2 m) (j + l - 2 n)}\varphi {\overset{m + 1,n}{\odot}}\sCurl \phi\nonumber\\
& + \tfrac{(-1)^{m + n + 1} (i -  m) (l -  n) (i + k -  m + 1)}{(i + 1) (i + k - 2 m) (j + l - 2 n)}\varphi {\overset{m,n + 1}{\odot}}\sCurlDagger \phi\nonumber\\
& + \tfrac{(-1)^{m + n} (i -  m) (j -  n) (i + k -  m + 1) (j + l -  n + 1)}{(i + 1) (j + 1) (i + k - 2 m) (j + l - 2 n)}\varphi {\overset{m,n}{\odot}}\sDiv \phi\nonumber\\
& + \tfrac{(i -  m) (j -  n)}{(i + k - 2 m) (j + l - 2 n)}\phi {\overset{m + 1,n + 1}{\odot}}\sTwist \varphi\nonumber\\
& + \tfrac{(- i + m) (l -  n) (j + l -  n + 1)}{(l + 1) (i + k - 2 m) (j + l - 2 n)}\phi {\overset{m + 1,n}{\odot}}\sCurl \varphi\nonumber\\
& + \tfrac{(j -  n) (- k + m) (i + k -  m + 1)}{(k + 1) (i + k - 2 m) (j + l - 2 n)}\phi {\overset{m,n + 1}{\odot}}\sCurlDagger \varphi\nonumber\\
& + \tfrac{(k -  m) (l -  n) (i + k -  m + 1) (j + l -  n + 1)}{(k + 1) (l + 1) (i + k - 2 m) (j + l - 2 n)}\phi {\overset{m,n}{\odot}}\sDiv \varphi 
\label{DivLeibniz}.
\end{align}
\end{subequations}
\end{lemma}
\begin{proof}
Collectively, the left hand sides can be written as $\nabla\overset{t,u}{\odot}(\phi\overset{m,n}{\odot}\varphi)$ where $t, u\in \{0,1\}$.
Let $\nabla_\phi$ and $\nabla_\varphi$ be $\nabla$ only differentiating $\phi$ respectively $\varphi$. From the relations \eqref{eq:Commutativity} and \eqref{eq:SymMultCommutator1} we get
\begin{align}
\nabla\overset{t,u}{\odot}(\phi\overset{m,n}{\odot}\varphi)
={}&\nabla_\phi\overset{t,u}{\odot}(\phi\overset{m,n}{\odot}\varphi)+\nabla_\varphi\overset{t,u}{\odot}(\phi\overset{m,n}{\odot}\varphi)\nonumber\\
={}&(-1)^{m+n}\nabla_\phi\overset{t,u}{\odot}(\varphi\overset{m,n}{\odot}\phi)+\nabla_\varphi\overset{t,u}{\odot}(\phi\overset{m,n}{\odot}\varphi)\nonumber\\
={}&(-1)^{m+n}\sum_{M=0}^{1}\sum_{N=0}^{1} F_{1,k,i}^{t,m,M}F_{1, l, j}^{u, n, N}
\varphi\overset{t+m-M,u+n-N}{\odot}\nabla\overset{M,N}{\odot}\phi\nonumber\\
&+\sum_{M=0}^{1}\sum_{N=0}^{1} F_{1,i,k}^{t,m,M}F_{1, j, l}^{u, n, N}
\phi\overset{t+m-M,u+n-N}{\odot}\nabla\overset{M,N}{\odot}\varphi.
\end{align}
Explicit calculations of the $F_{1,i,k}^{t,m,M}$ coefficients gives the relations \eqref{eq:SymMultLeibniz}.
\end{proof}

\subsection{GHP expansion}\label{sec:GHP}

In this section we collect equations to efficiently expand symmetric spinorial equations into GHP components. Let us first briefly review the formalism, see \cite{GHP} for details. Introducing a normalized spinor dyad $(o_A, \iota_A), o_A \iota^A=1$, a two dimensional subgroup of the Lorentz group is given by
\begin{align} \label{GHPDyadTrafo}
o_A \to \lambda o_A, \qquad \iota_A \to \lambda^{-1} \iota_A,
\end{align}
with non-vanishing, complex scalar field $\lambda$. A field $\phi$ is said to be of GHP weight $\{p,q\}$ if it transforms via
\begin{align}
 \phi \to \lambda^p \bar \lambda^q \phi
\end{align}
under \eqref{GHPDyadTrafo} and its complex conjugate. The Levi-Civita connection has a natural lift of the form
\begin{align} \label{ThetaDef}
\Theta_{AA'} = \nabla_{AA'} - p \omega_{AA'} - q \bar{\omega}_{AA'}, \qquad \text{with } \omega_{AA'} = \iota^B \nabla_{AA'} o_B,
\end{align}
and is of weight zero in the sense that it maps $\{p,q\}$ weighted fields to $\{p,q\}$ fields. The GHP operators are given by the dyad expansion of \eqref{ThetaDef},
\begin{align} \label{ThetaDyad}
\Theta_{AA'} ={}& \iota_A \bar \iota_{A'} \tho - \iota_A \bar o_{A'} \edt - o_A \bar \iota_{A'} \edt' + o_A \bar o _{A'} \tho'.
\end{align}
The connection coefficients are defined as follows,
\begin{subequations} \label{GammaDef}
\begin{align}
\Theta_{AA'} o_B ={}& \Gamma_{AA'} \iota_B, \qquad \text{where }
\Gamma_{AA'} = -\iota_A \bar \iota_{A'} \kappa + \iota_A \bar o_{A'} \sigma + o_A \bar \iota_{A'} \rho - o_A \bar o _{A'} \tau, \\
\Theta_{AA'} \iota_B ={}& \Gamma'_{AA'} o_B, \qquad \text{where } 
\Gamma'_{AA'} = -\iota_A \bar \iota_{A'} \tau' + \iota_A \bar o_{A'} \rho' + o_A \bar \iota_{A'} \sigma' - o_A \bar o _{A'} \kappa'.
\end{align}
\end{subequations}

To express the dyad expansion of a general symmetric spinor, it is convenient to define a symmetric spinor basis $\SymDyadBasis{n}{k}{m}{l}$ of weight $\{2n-k,2m-l\}$ by
\begin{align}
\SymDyadBasis{n}{k}{m}{l}{}_{A_1\dots A_k}^{A'_1\dots A'_l} &= o_{(A_1} \dots o_{A_n} \iota_{A_{n+1}} \dots \iota_{A_k)} \bar o^{(A'_1} \dots \bar o^{A'_m} \bar \iota^{A'_{m+1}} \dots \bar \iota^{A'_l)}.
\end{align}
In particular this allows us to mostly avoid spinor indices for the rest of this section. For a full contraction of two basis elements we find
\begin{align} \label{BasisFullContraction}
\SymDyadBasis{n}{k}{m}{l} \overset{k,l}{\odot} \SymDyadBasis{i}{k}{j}{l} = 
(-1)^{(n+m)} \binom{k}{n}^{-1} \binom{l}{m}^{-1} \delta^{n}_{k-i} \delta^m_{l-j},
\end{align}
where $\delta^a_b=1$ if $a=b$ and zero otherwise. Now any $\phi \in \SymSpace_{k,l}$ can be expanded into
\begin{align} \label{eq:phiklGHPExpansion}
\phi
={}& \sum_{i=0}^k \sum_{j=0}^l (-1)^{k-i+l-j} \binom{k}{i}\binom{l}{j} \phi_{ij'} \SymDyadBasis{i}{k}{j}{l},
\end{align}
where the scalar components of weight $\{k-2i,l-2j\}$ are defined by
\begin{align} \label{eq:phiklGHPComponents}
\phi_{ij'} ={}& \SymDyadBasis{k-i}{k}{l-j}{l} \overset{k,l}{\odot}\phi.
\end{align}

The following two lemmas yield component expressions for general symmetric products and derivatives of symmetric spinors. This allows to expand general symmetric spinor differential equations into dyad components, without expanding the symmetrizations.

\begin{lemma}
For $\phi \in \SymSpace_{i,j}, \varphi \in \mathcal{S}_{k,l}$ the symmetric product has components
\begin{align} \label{eq:SymSpinGHP}
(\phi\overset{m,n}{\odot}\varphi)_{st'} 
={}&\sum_{p=0}^{k}\sum_{q=0}^{l}G_{i,k}^{m,p,s}G_{j,l}^{n,q,t}\phi_{(s+m-p)(t+n-q)'} \varphi_{pq'},
\end{align}
with coefficients given by
\begin{align}
G_{i,k}^{m,p,s}={}&\sum_{r=0}^m (-1)^{r}\frac{\binom{i}{s+m-p} \binom{i +p-s-m}{r} \binom{s+m-p}{m-r}\binom{k}{p} \binom{k -p}{m -  r} \binom{p}{r}}{\binom{i}{m}\binom{k}{m}\binom{m}{r}\binom{i+k-2m}{s}}.
\end{align}
\end{lemma}
\begin{proof}
For ease of notation we assume $\phi \in \SymSpace_{i,0}, \varphi \in \mathcal{S}_{k,0}$. 
Using the observation that  $\SymDyadBasis{p}{k}{0}{0}=\SymDyadBasis{0}{k-p}{0}{0}\overset{0,0}{\odot}
\SymDyadBasis{p}{p}{0}{0}$, where $\SymDyadBasis{0}{k-p}{0}{0}$ is a symmetric product of $\iota_A$ and $\SymDyadBasis{p}{p}{0}{0}$ is a symmetric product of $o_A$, we can use \eqref{eq:PartialSymMultExpansion} to obtain
\begin{align}
\SymDyadBasis{p}{k}{0}{0}{}_{A_1\dots A_{k-m}}^{B_1\dots B_m}={}&(\SymDyadBasis{0}{k-p}{0}{0}\overset{0,0}{\odot}
\SymDyadBasis{p}{p}{0}{0}){}_{A_1\dots A_{k-m}}^{B_1\dots B_m}\nonumber\\
={}&\sum_{q=0}^{m}\frac{\binom{p}{q} \binom{k - p}{m-q}}{\binom{k}{m}}
 \SymDyadBasis{0}{k-p}{0}{0}{}^{(B_{q+1}\dots B_{m}}_{(A_1\dots A_{k-p-m+q}}\SymDyadBasis{p}{p}{0}{0}{}_{A_{k-p-m+q+1}\dots A_{k-m})}^{B_1\dots B_q)}\nonumber\\
 ={}&\sum_{q=0}^{m}\frac{\binom{p}{q} \binom{k - p}{m-q}}{\binom{k}{m}}
 \SymDyadBasis{p-q}{k-m}{0}{0}{}_{A_1\dots  A_{k-m}}\SymDyadBasis{q}{m}{0}{0}{}^{B_1\dots B_{m}}
\end{align}
Using this in the expansion \eqref{eq:phiklGHPExpansion}, we find
\begin{subequations}
\begin{align}
\varphi_{A_1 \dots A_{k-m}B_1\dots B_m} 
={}& \sum_{p=0}^k \sum_{q=0}^m (-1)^{k-p}\frac{\binom{k}{p} \binom{k -  p}{m -  q} \binom{p}{q}}{\binom{k}{m}} \varphi_{p0'} 
\SymDyadBasis{p-q}{k-m}{0}{0}{}_{A_1\dots A_{k-m}}
\SymDyadBasis{q}{m}{0}{0}{}_{B_1\dots B_m},\\
\phi_{A_1 \dots A_{i-m}}^{B_1\dots B_m} 
={}& \sum_{r=0}^i \sum_{q=0}^m (-1)^{i-r}\frac{\binom{i}{r} \binom{i - r}{q} \binom{r}{m-q}}{\binom{i}{m}} \phi_{r0'} 
\SymDyadBasis{r-m+q}{i-m}{0}{0}{}_{A_1\dots A_{i-m}}
\SymDyadBasis{q}{m}{0}{0}{}^{B_1\dots B_{m}}.
\end{align}
\end{subequations}
Contracting the $B$ indices, symmetrizing and using \eqref{BasisFullContraction} yield
\begin{align}
&\hspace{-3ex}\phi_{(A_1 \dots A_{i-m}}^{B_1\dots B_m} \varphi_{A_{i-m+1} \dots A_{i+k-2m})B_1\dots B_m} \nonumber\\
={}& \sum_{r=0}^i\sum_{p=0}^k \sum_{q=0}^m (-1)^{k+i-r-p}\frac{\binom{i}{r} \binom{i - r}{q} \binom{r}{m-q}\binom{k}{p} \binom{k -  p}{m -  q} \binom{p}{q}}{\binom{i}{m}\binom{k}{m}}\phi_{r0'} \varphi_{p0'} 
\SymDyadBasis{p+r-m}{i+k-2m}{0}{0}{}_{A_1\dots A_{i+k-2m}}
\nonumber\\
&\times\SymDyadBasis{q}{m}{0}{0}{}_{B_1\dots B_{m}}\SymDyadBasis{q}{m}{0}{0}{}^{B_1\dots B_{m}}\nonumber\\
={}& \sum_{r=0}^i\sum_{p=0}^k \sum_{q=0}^m (-1)^{k+i-r-p+m-q}\frac{\binom{i}{r} \binom{i - r}{q} \binom{r}{m-q}\binom{k}{p} \binom{k -  p}{m -  q} \binom{p}{q}}{\binom{i}{m}\binom{k}{m}\binom{m}{q}}\phi_{r0'} \varphi_{p0'} 
\SymDyadBasis{p+r-m}{i+k-2m}{0}{0}{}_{A_1\dots A_{i+k-2m}}.
\end{align}
The relation \eqref{BasisFullContraction} then gives
\begin{align}
(\phi\overset{m,0}{\odot}\varphi)_{s0'} ={}& \SymDyadBasis{i+k-2m-s}{i+k-2m}{0}{0} \overset{i+k-2m,0}{\odot}(\phi\overset{m,0}{\odot}\varphi)\nonumber\\
={}&\sum_{r=0}^i\sum_{p=0}^k \sum_{q=0}^m (-1)^{-r-p+m-q-s}\frac{\binom{i}{r} \binom{i - r}{q} \binom{r}{m-q}\binom{k}{p} \binom{k -  p}{m -  q} \binom{p}{q}}{\binom{i}{m}\binom{k}{m}\binom{m}{q}\binom{i+k-2m}{i+k-2m-s}}\phi_{r0'} \varphi_{p0'} 
   \delta^{m+s-p}_{r}\nonumber\\
={}&\sum_{p=0}^k \sum_{q=0}^m (-1)^{q}\frac{\binom{i}{m+s-p} \binom{i + p - m - s}{q} \binom{m+s-p}{m-q}\binom{k}{p} \binom{k -  p}{m -  q} \binom{p}{q}}{\binom{i}{m}\binom{k}{m}\binom{m}{q}\binom{i+k-2m}{i+k-2m-s}}\phi_{(m+s-p)0'} \varphi_{p0'} \nonumber\\
={}&\sum_{p=0}^kG_{i,k}^{m,p,s}\phi_{(s+m-p)0'} \varphi_{p0'}.
\end{align}
The primed indices gives an analogous expansion and the combination yields \eqref{eq:SymSpinGHP}. 
\end{proof}

\begin{lemma}
The GHP components of fundamental spinor operators \eqref{eq:FundSpinOps} on $\phi \in \SymSpace_{k,l}$ take the form
\begin{subequations} 
\begin{align} \label{DivGHP}
(\sDiv\phi)_{ij'} ={}&
(\tho - (k-i) \rho -  (l-j) \bar\rho)\phi_{(i+1)(j+1)'} 
+ (\tho' - (i+1) \rho' - (j+1) \bar\rho')\phi_{ij'} \nonumber \\
&- (\edt - (k-i)\tau - (j+1) \bar{\tau}')\phi_{(i+1)j'} 
- (\edt' - (i+1) \tau' - (l-j)\bar{\tau})\phi_{i(j+1)'} \nonumber \\
& + (k-i-1) \kappa \phi_{(i+2)(j+1)'}
- (k-i-1) \sigma \phi_{(i+2)j'}
- i \sigma' \phi_{(i-1)(j+1)'} \nonumber \\
&+ i \kappa' \phi_{(i-1)j'}
+ (l-j-1) \bar\kappa \phi_{(i+1)(j+2)'}
- (l-j-1) \bar\sigma \phi_{i(j+2)'} \nonumber \\
&- j \bar\sigma' \phi_{(i+1)(j-1)'}
+ j \bar\kappa' \phi_{i(j-1)'}, \\
(\sCurl\phi)_{ij'} ={}& \bigl(
-(k-i+1)(\tho + i \rho -  (l-j) \bar\rho)\phi_{i(j+1)'} 
+i(\tho' +(k-i+1)\rho' -(j+1) \bar\rho')\phi_{(i-1)j'} \nonumber \\
&\hspace{1ex}{} + (k-i+1)(\edt + i \tau -  (j+1) \bar\tau')\phi_{ij'} 
-i(\edt' + (k-i+1)\tau'-(l-j) \bar\tau)\phi_{(i-1)(j+1)'} \nonumber \\
&\hspace{1ex}{} - (k-i+1)(k-i) \kappa \phi_{(i+1)(j+1)'} 
+(k-i+1)(k-i) \sigma \phi_{(i+1)j'} 
- i(i-1)  \sigma' \phi_{(i-2)(j+1)'}\nonumber \\
&\hspace{1ex}{}+ i(i-1) \kappa' \phi_{(i-2)j'}
-(k-i+1)(l-j-1) \bar\kappa \phi_{i(j+2)'}
-i(l-j-1) \bar\sigma \phi_{(i-1)(j+2)'}\nonumber \\
&\hspace{1ex}{}+ (k-i+1)j \bar\sigma' \phi_{i(j-1)'}
+ij \bar\kappa' \phi_{(i-1)(j-1)'} \bigr) / (k+1), \\
(\sCurlDagger\phi)_{ij'} ={}&  \bigl(
-(l-j+1)(\tho + j \bar\rho -  (k-i) \rho)\phi_{(i+1)j'} 
+j(\tho' +(l-j+1)\bar\rho' -(i+1) \rho')\phi_{i(j-1)'} \nonumber \\
&\hspace{1ex}{} + (l-j+1)(\edt' + j \bar\tau -  (i+1) \tau')\phi_{ij'} 
-j(\edt + (l-j+1)\bar\tau'-(k-i) \tau)\phi_{(i+1)(j-1)'} \nonumber \\
&\hspace{1ex}{}- (l-j+1)(l-j) \bar\kappa \phi_{(i+1)(j+1)'} 
+(l-j+1)(l-j) \bar\sigma \phi_{i(j+1)'} 
- j(j-1)  \bar\sigma' \phi_{(i+1)(j-2)'}\nonumber \\
&\hspace{1ex}{}+ j(j-1) \bar\kappa' \phi_{i(j-2)'}
-(l-j+1)(k-i-1) \kappa \phi_{(i+2)j'}
-j(k-i-1) \sigma \phi_{(i+2)(j-1)'}\nonumber \\
&\hspace{1ex}{}+ (l-j+1)i \sigma' \phi_{(i-1)j'}
+ij \kappa' \phi_{(i-1)(j-1)'} \bigr) / (l+1), \\
(\sTwist\phi)_{ij'} ={}& \bigl(
(k+1-i)(l+1-j)(\tho + i \rho +  j\bar\rho)\phi_{ij'} \nonumber \\ 
&\hspace{1ex}{}+(k+1-i)j(\edt + i\tau + (l-j+1)\bar\tau')\phi_{i(j-1)'} \nonumber \\ 
&\hspace{1ex}{}+i(l+1-j)(\edt' + (k-i+1)\tau' +  j\bar\tau)\phi_{(i-1)j'} \nonumber \\ 
&\hspace{1ex}{}+ij(\tho' +(k-i+1)\rho' +(l-j+1)\bar\rho')\phi_{(i-1)(j-1)'}  \nonumber \\ 
&\hspace{1ex}{}+(k-i+1)(k-i)(l+1-j) \kappa \phi_{(i+1)j'} 
  + (k-i+1)(k-i)j \sigma \phi_{(i+1)(j-1)'} \nonumber\\
&\hspace{1ex}{}+i(i-1)(l+1-j) \sigma' \phi_{(i-2)j'}
  +i(i-1)j \kappa' \phi_{(i-2)(j-1)'} \nonumber\\
&\hspace{1ex}{}+(k+1-i)(l+1-j)(l-j) \bar\kappa \phi_{i(j+1)'} 
  + i(l+1-j)(l-j) \bar\sigma \phi_{(i-1)(j+1)'}\nonumber\\
&\hspace{1ex}{}+(k+1-i)j(j-1) \bar\sigma' \phi_{i(j-2)'} 
  +ij(j-1) \bar\kappa' \phi_{(i-1)(j-2)'} \big) / \big((k+1)(l+1)\bigr)
\end{align}
\end{subequations} 
\end{lemma}

\begin{proof} 
To prove \eqref{DivGHP}, we start by expanding the argument of $\sDiv\phi$ using \eqref{eq:phiklGHPExpansion} and contract with a symmetric basis as in \eqref{eq:phiklGHPComponents},
\begin{align} \label{DivphiStep1}
(\sDiv\phi)_{ij'} 
={}& \SymDyadBasis{k-1-i}{k-1}{l-1-j}{l-1} \overset{k,l}{\odot} (\sDiv\phi) \nonumber \\
={}& \sum_{n=0}^k \sum_{m=0}^l (-1)^{k-n+l-m} \binom{k}{n}\binom{l}{m} \SymDyadBasis{k-1-i}{k-1}{l-1-j}{l-1} \overset{k,l}{\odot} (\sDiv(\phi_{nm'}  \SymDyadBasis{n}{k}{m}{l}) ).
\end{align}
Next, we use the Leibniz rule \eqref{DivLeibniz}, but switch to the GHP connection $\Theta_{AA'}$ (so the fundamental spinor operators are with respect to $\Theta_{AA'}$ instead of $\nabla_{AA'}$) as the GHP components and the basis elements are GHP weighted,
\begin{align}\label{DivphiBasis}
\sDiv( \phi_{nm'}\underset{k,l}{\overset{0,0}{\odot}}\SymDyadBasis{n}{k}{m}{l})={}&\SymDyadBasis{n}{k}{m}{l} \overset{1,1}{\odot}\sTwist \phi_{nm}
 + \phi_{nm}\overset{0,0}{\odot}\sDiv \SymDyadBasis{n}{k}{m}{l}.
\end{align}
From \eqref{ThetaDyad} and \eqref{GammaDef} we have
\begin{align}\label{Twistphinm}
\sTwist\phi_{nm'} ={}& (\tho \phi_{nm'}) \SymDyadBasis{0}{1}{0}{1} - (\edt \phi_{nm'}) \SymDyadBasis{0}{1}{1}{1} - (\edt' \phi_{nm'}) \SymDyadBasis{1}{1}{0}{1} + (\tho' \phi_{nm'}) \SymDyadBasis{1}{1}{1}{1},
\end{align}
and
\begin{align} \label{DivBasis}
\sDiv \SymDyadBasis{n}{k}{m}{l} 
={}& n \Gamma \overset{1,1}{\odot} \SymDyadBasis{n-1}{k}{m}{l} 
+ (k-n) \Gamma' \overset{1,1}{\odot} \SymDyadBasis{n+1}{k}{m}{l}  +m \overline{\Gamma} \overset{1,1}{\odot} \SymDyadBasis{n}{k}{m-1}{l} 
+(l-m) \overline{\Gamma'} \overset{1,1}{\odot} \SymDyadBasis{n}{k}{m+1}{l}.  
\end{align}
Inserting \eqref{Twistphinm}, \eqref{DivBasis} back into \eqref{DivphiBasis} and expanding $\Gamma, \Gamma'$ into the basis we can use the contraction rules
\begin{subequations} \label{Basis11vec}
\begin{align}
\SymDyadBasis{n}{k}{m}{l} \overset{1,1}{\odot}\SymDyadBasis{0}{1}{0}{1} ={}&\frac{m n}{k l}\SymDyadBasis{n-1}{k-1}{m-1}{l-1}, \\
\SymDyadBasis{n}{k}{m}{l} \overset{1,1}{\odot}\SymDyadBasis{1}{1}{0}{1}={}&- \frac{m (k -  n)}{k l} \SymDyadBasis{n}{k-1}{m-1}{l-1}, \\
\SymDyadBasis{n}{k}{m}{l} \overset{1,1}{\odot}\SymDyadBasis{0}{1}{1}{1}={}&- \frac{n (l -  m)}{k l} \SymDyadBasis{n-1}{k-1}{m}{l-1}, \\
\SymDyadBasis{n}{k}{m}{l} \overset{1,1}{\odot}\SymDyadBasis{1}{1}{1}{1}={}&\frac{(k -  n) (l -  m)}{k l}\SymDyadBasis{n}{k-1}{m}{l-1},
\end{align}
\end{subequations}
which are easily verified by expanding out the symmetries. The result can now be substituted into \eqref{DivphiStep1}. Each term has a full contraction of the form \eqref{BasisFullContraction} which cancels the double sum due to the $\delta$ factors. After some elementary algebra, the end result is given by \eqref{DivGHP}. The other expansions can be verified along the same lines, the only minor computation that needs to be done is the analog of \eqref{DivBasis} and \eqref{Basis11vec}.
\end{proof}

\section{SymSpin: A computer algebra implementation in xAct}\label{sec:SymSpinPackage} 

The \emph{xAct} \cite{xActWeb} suite for \emph{Mathematica} is an open source project mainly devoted to symbolic computation in differential geometry and tensor algebra. In this section we introduce our contributed package \emph{SymSpin} \cite{SymSpinWeb} which contains the formalism of Section~\ref{sec:SymSpinAlg}. For syntax and more examples, see \emph{SymSpinDoc.nb} on that page.

\subsection{Loading the package and defining structures}
Load the package, define a four dimensional manifold M4, and Lorentzian metric  with
\begin{mma}
	\mmain{<$\mspace{0.5mu}$<xAct\textasciigrave SymSpin\textasciigrave\linebreak
	\$DefInfoQ=False;\linebreak
	DefManifold[M4,4,\{a,b,c,d\}]\linebreak
	DefMetric[\{1,3,0\},g[-a,-b],CD]}
\end{mma}%
By default the valence numbers are displayed for each operator and complex conjugates are written with $\dagger$. To keep the notation the same as in the rest of the paper, we can change the display form with
\begin{mma}
	\mmain{SetOptions[DefAbstractIndex,PrintAs->PrimeDagger];\linebreak
	SetOptions[DefSpinor,~PrintDaggerAs->AddBar];\linebreak
	SetOptions[DefFundSpinOperators,ShowValenceInfo->False];} 
\end{mma}%
Define the spin structure, initialize \emph{SymSpin} and define the fundamental spinor operators with
\begin{mma}
	\mmain{DefSpinStructure[g,Spin,\{A,B,C,F,G,H,P,Q,R\},$\epsilon $,$\sigma $,CDe,\{";","$\nabla$"\},\linebreak
	SpinorPrefix->SP,SpinorMark->"S"]\linebreak
	InitSymSpin[$\sigma$];\linebreak
	DefFundSpinOperators[CDe];}
\end{mma}%

\subsection{Example: Coefficients} \label{sec:example}
Assume that $K$, $L$ and $M$ are symmetric spinor fields, and we want to find under which conditions of $K$, $L$ and $M$ the equation
\begin{align}
0={}&K_{AB}{}^{FH} L_{F}{}^{C} \varphi_{HC}
 + M_{(A}{}^{C}\varphi_{B)C}.
\end{align}
holds for all symmetric spinor fields $\varphi$. 
The following calculation leads to the conditions
\begin{align}
K^{G}{}_{(ABC}L_{|G|F)}={}&0,&
M_{AB}={}&\tfrac{1}{2} K^{CF}{}_{AB} L_{CF}.
\end{align}

We first define the symmetric spinor fields. For clarity we have added the valence numbers to the names of the spinors, but not the display form.
\begin{mma}
	\mmain{DefSymmetricSpinor[$\varphi $20,2,0,Spin,"$\varphi $"]\linebreak
	DefSymmetricSpinor[K40,4,0,Spin,"K"]\linebreak
	DefSymmetricSpinor[L20,2,0,Spin,"L"]\linebreak
	DefSymmetricSpinor[M20,2,0,Spin,"M"]}
\end{mma}%
One can start with the indexed version of the spinor equation.
\begin{mma}
	\mmain{OriginalEq=0==K40[-A,-B,F,H]L20[-F,C]$\varphi $20[-H,-C]\linebreak+ImposeSym[M20[-A,C]*$\varphi $20[-B,-C]]}\\ 
	\mmaout{0~\texttt{==}~K_{AB}{}^{FH} L_{F}{}^{C} \varphi_{HC} + \underset{\scriptscriptstyle(13)}{Sym}[M\varphi ]_{A}{}^{C}{}_{BC}}
\end{mma}%
To convert this to the new formalism, we need the irreducible decomposition of the product of the $L$ and $\phi$ spinor. 
\begin{mma}
	\mmain{IrrDecomposeSymMult[L20,$\varphi $20,\{0,0\}]}\\ 
	\mmaout{L_{AB} \varphi_{CF}~\texttt{==}~- \underset{\scriptscriptstyle(13)(24)}{Sym}[\epsilon (L{\overset{1,0}{\odot}}\varphi)]_{ACBF} + (L{\overset{0,0}{\odot}}\varphi)_{ABCF} + \tfrac{1}{3} \underset{\scriptscriptstyle(13)(24)}{Sym}[\epsilon \epsilon ]_{ACBF} (L{\overset{2,0}{\odot}}\varphi)}
\end{mma}%
It is convenient to work with the expanded and canonicalized version
\begin{mma}
	\mmain{L20$\varphi $20IrrDecEq=ToCanonical@ExpandSym@\%}\\ 
	\mmaout{L_{AB} \varphi_{CF}~\texttt{==}~(L{\overset{0,0}{\odot}}\varphi)_{ABCF} -  \tfrac{1}{4} \epsilon_{BF} (L{\overset{1,0}{\odot}}\varphi)_{AC} -  \tfrac{1}{4} \epsilon_{BC} (L{\overset{1,0}{\odot}}\varphi)_{AF} -  \tfrac{1}{4} \epsilon_{AF} (L{\overset{1,0}{\odot}}\varphi)_{BC}}\\ 
	\mmanoout{ -  \tfrac{1}{4} \epsilon_{AC} (L{\overset{1,0}{\odot}}\varphi)_{BF} + \tfrac{1}{6} \epsilon_{AF} \epsilon_{BC} (L{\overset{2,0}{\odot}}\varphi) + \tfrac{1}{6} \epsilon_{AC} \epsilon_{BF} (L{\overset{2,0}{\odot}}\varphi)}
\end{mma}%

To work efficiently we turn the original equation into an index-free version. One could also use the index-free version as a starting point.
\begin{mma}
	\mmain{IndexFreeEq=ToIndexFree[ToCanonical@ContractMetric[OriginalEq\linebreak
	/.EqToRule@L20$\varphi $20IrrDecEq]//.SymHToSymMultRule]\linebreak
	/.MultScalToSymMultRule[Spin]/.SortSymMult[Not@FreeQ[\#,$\varphi $20]\&]}\\ 
	\mmaout{0~\texttt{==}~M{\overset{1,0}{\odot}}\varphi + K{\overset{2,0}{\odot}}L{\overset{1,0}{\odot}}\varphi }
\end{mma}%
We can turn the spinor valued equation into a scalar equation by contracting it with a dummy spinor $T$ to turn the free indices into contracted dummy indices. This dummy spinor is defined by
\begin{mma}
	\mmain{DefSymmetricSpinor[T20,2,0,Spin,"T"]}
\end{mma}%
As the field $\varphi$ and the dummy spinor $T$ both should be arbitrary, we see that the irreducible components of their product can be treated as independent arbitrary fields. For convenience we make a list of them with
\begin{mma}
	\mmain{IrrDecComps=SymMult[T20,\#,0,Spin][$\varphi $20]\&/@Range[0,2]}\\ 
	\mmaout{\texttt{\{}(T{\overset{0,0}{\odot}}\varphi)\texttt{, }(T{\overset{1,0}{\odot}}\varphi)\texttt{, }(T{\overset{2,0}{\odot}}\varphi)\texttt{\}}}
\end{mma}%
We can now contract our index-free equation with $T$.
\begin{mma}
	\mmain{SymMult[T20,2,0]/@IndexFreeEq}\\ 
	\mmaout{0~\texttt{==}~T{\overset{2,0}{\odot}}M{\overset{1,0}{\odot}}\varphi + T{\overset{2,0}{\odot}}K{\overset{2,0}{\odot}}L{\overset{1,0}{\odot}}\varphi }
\end{mma}%
Commute $T$ inside, so that $T$ is directly contracted with the field $\varphi$, so we obtain the independent spinors in the list \texttt{IrrDecComps}.
\begin{mma}
	\mmain{\%//.CommuteSymMultRuleIn[T20]}\\ 
	\mmaout{0~\texttt{==}~- M{\overset{2,0}{\odot}}T{\overset{1,0}{\odot}}\varphi + K{\overset{4,0}{\odot}}L{\overset{1,0}{\odot}}T{\overset{0,0}{\odot}}\varphi + \tfrac{1}{2}K{\overset{4,0}{\odot}}L{\overset{0,0}{\odot}}T{\overset{1,0}{\odot}}\varphi }
\end{mma}%
Now, these independent spinors are moved out and to the left.
\begin{mma}
	\mmain{\%/.SortSymMultReverse[MemberQ[IrrDecComps,\#]\&]\linebreak
	//.Flatten[CommuteSymMultRuleOut/@IrrDecComps]}\\ 
	\mmaout{0~\texttt{==}~- (T{\overset{1,0}{\odot}}\varphi){\overset{2,0}{\odot}}M + (T{\overset{0,0}{\odot}}\varphi){\overset{4,0}{\odot}}K{\overset{1,0}{\odot}}L + \tfrac{1}{2}(T{\overset{1,0}{\odot}}\varphi){\overset{2,0}{\odot}}K{\overset{2,0}{\odot}}L}
\end{mma}%
From this one can conclude that the coefficients of $(T\overset{1,0}{\odot}\varphi)$ and $(T\overset{0,0}{\odot}\varphi)$ both have to be zero. 

As a convenience, we have implemented all of the steps from the index-free equation to the final list of equations in one function.
\begin{mma}
	\mmain{ExtractCoeffsIndexFree[IndexFreeEq,$\varphi $20]}\\ 
	\mmaout{\texttt{\{}0~\texttt{==}~(K{\overset{1,0}{\odot}}L)\texttt{, }0~\texttt{==}~- M + \tfrac{1}{2}K{\overset{2,0}{\odot}}L\texttt{\}}}
\end{mma}%
This can be translated back to the indexed form with
\begin{mma}
	\mmain{ToIndexed/@\%}\\ 
	\mmaout{\texttt{\{}0~\texttt{==}~\underset{\scriptscriptstyle(2346)}{Sym}[KL]^{G}{}_{ABCGF}\texttt{, }0~\texttt{==}~\tfrac{1}{2} K^{CF}{}_{AB} L_{CF} -  M_{AB}\texttt{\}}}
\end{mma}%
Performing this kind of calculation in the indexed form would require expansions of symmetries and several steps of irreducible decompositions of different products. This new method was heavily used in \cite{JacBac2022}.

\subsection{Example: Derivatives}
To also demonstrate how to work with derivatives we use the previously defined field $\varphi$ and define a valence $(3,2)$ field $\psi$ via
\begin{mma}
	\mmain{DefSymmetricSpinor[$\psi $32,3,2,Spin,"$\psi $"]}
\end{mma}%
The covariant derivative
\begin{mma}
	\mmain{CDe[-A,-A$\dagger $]@$\psi $32[-B,-C,-F,-B$\dagger $,-C$\dagger $]}\\ 
	\mmaout{\nabla_{AA'}\psi_{BCFB'C'}}
\end{mma}%
can be decomposed into the fundamental spinor operators with
\begin{mma}
	\mmain{\%==ToFundSpinOp[\%]}\\ 
	\mmaout{\nabla_{AA'}\psi_{BCFB'C'}~\texttt{==}~- \tfrac{1}{3} \bar{\epsilon}_{A'C'} (\sCurl \psi)_{ABCFB'} -  \tfrac{1}{3} \bar{\epsilon}_{A'B'} (\sCurl \psi)_{ABCFC'} -  \tfrac{1}{4} \epsilon_{AF} (\sCurlDagger \psi)_{BCA'B'C'}}\\ 
	\mmanoout{ -  \tfrac{1}{4} \epsilon_{AC} (\sCurlDagger \psi)_{BFA'B'C'} -  \tfrac{1}{4} \epsilon_{AB} (\sCurlDagger \psi)_{CFA'B'C'} + \tfrac{1}{12} \epsilon_{AF} \bar{\epsilon}_{A'C'} (\sDiv \psi)_{BCB'}}\\ 
	\mmanoout{ + \tfrac{1}{12} \epsilon_{AF} \bar{\epsilon}_{A'B'} (\sDiv \psi)_{BCC'} + \tfrac{1}{12} \epsilon_{AC} \bar{\epsilon}_{A'C'} (\sDiv \psi)_{BFB'} + \tfrac{1}{12} \epsilon_{AC} \bar{\epsilon}_{A'B'} (\sDiv \psi)_{BFC'}}\\ 
	\mmanoout{ + \tfrac{1}{12} \epsilon_{AB} \bar{\epsilon}_{A'C'} (\sDiv \psi)_{CFB'} + \tfrac{1}{12} \epsilon_{AB} \bar{\epsilon}_{A'B'} (\sDiv \psi)_{CFC'} + (\sTwist \psi)_{ABCFA'B'C'}}
\end{mma}%
Commutators can be handled like
\begin{mma}
	\mmain{DivCDe@CurlDgCDe@$\psi $32}\\ 
	\mmaout{(\sDiv \sCurlDagger \psi)}
\end{mma}%
\begin{mma}
	\mmain{\%==(\%/.CommuteOp[DivCDe,CurlDgCDe])}\\ 
	\mmaout{(\sDiv \sCurlDagger \psi)~\texttt{==}~\tfrac{2}{3}\sCurlDagger \sDiv \psi + \Psi {\overset{3,0}{\odot}}\psi + 2\Phi {\overset{2,1}{\odot}}\psi }
\end{mma}%
Derivatives of products can also be handled efficiently
\begin{mma}
	\mmain{CurlDgCDe@SymMult[$\varphi $20,1,0]@$\psi $32}\\ 
	\mmaout{(\sCurlDagger \varphi {\overset{1,0}{\odot}}\psi)}
\end{mma}%
\begin{mma}
	\mmain{\%$== $(\%/.SymMultLeibnizRules[CDe])}\\ 
	\mmaout{(\sCurlDagger \varphi {\overset{1,0}{\odot}}\psi)~\texttt{==}~\tfrac{2}{3}\psi {\overset{2,0}{\odot}}\sTwist \varphi -  \tfrac{5}{9}\psi {\overset{1,0}{\odot}}\sCurlDagger \varphi -  \tfrac{1}{3}\varphi {\overset{2,0}{\odot}}\sTwist \psi + \tfrac{5}{6}\varphi {\overset{1,0}{\odot}}\sCurlDagger \psi }
\end{mma}%

\section{Conclusions and discussion}\label{sec:Conclusions}

In this work, we introduced an algebra on symmetric 2-spinors and the corresponding \emph{SymSpin} package for the \emph{Mathematica} suite \emph{xAct}. In various research projects of the authors this algebra turned out to be a very efficient way to perform calculations.
For example in \cite{JacBac2022} it is used to derive conditions on the spacetime for the existence of second order symmetry operators for the massive Dirac equation. This greatly simplified the calculations compared to the earlier approach \cite{AndBaeBlu14a}, where only parts of the formalism were used to investigate symmetry operators for the massless Dirac and the Maxwell equations.

The formalism is very efficient for cases where each spinor appears only once in each product. Choosing a preferred ordering of the factors in each product, one can use the relations in Theorem \ref{thm:SymProdProperties} to rewrite them in a canonical form.
However, if a spinor appears multiple times in a product the relations in Theorem \ref{thm:SymProdProperties} can give non-trivial equations where a term of the same form can appear both in the left and right hand sides as well as in several equations. Solving these equations, it should be possible to develop a method to write such products in a canonical form. We plan to continue the development of these tools for such cases in the future.

\subsection*{Acknowledgements}
The authors are grateful to Simon Jacobsson for testing of the \emph{xAct} implementation and to Teake Nutma for \LaTeX{}  typesetting code of \emph{Mathematica} expressions.


%

\pagebreak

\end{document}